\documentclass{article}

\usepackage{PRIMEarxiv}

\usepackage{graphicx}%
\usepackage{multirow}%
\usepackage{amsmath,amssymb,amsfonts}%
\usepackage{mathrsfs}%
\usepackage[title]{appendix}%
\usepackage{xcolor}%
\usepackage{textcomp}%
\usepackage{manyfoot}%
\usepackage{booktabs}%
\usepackage{algorithm}%
\usepackage{algorithmicx}%
\usepackage{algpseudocode}%
\usepackage{listings}%
\usepackage{comment}%
\usepackage{subcaption}

\usepackage{amsmath,amssymb}
\usepackage{bbm}
\usepackage{comment}
\usepackage{amsthm}
\usepackage{enumitem}

\newcommand{\Rb}{\mathbbm{R}}      % for Real numbers

\newcommand{\Bc}{\mathcal{B}}

\newcommand{\Qc}{\mathcal{Q}}

\newcommand{\Eb}{\mathbbm{E}}

\newcommand{\Pc}{\mathcal{P}}

\newcommand{\Xc}{\mathcal{X}}
\newcommand{\Yc}{\mathcal{Y}}

\newcommand{\Wc}{\mathcal{W}}

\newcommand{\1}{\mathbbm{1}}

\newcommand{\D}{\mathop{\text{\rm d}\!}}

\newcommand{\BLUE}[1]{{\color{blue} #1}}

    \definecolor{plum}{rgb}{0.3,0,0.7}

\newtheorem{theorem}{Theorem}%  meant for continuous numbers
\newtheorem{remark}{Remark}%
\newtheorem{lemma}{Lemma}%
\newtheorem{corollary}{Corollary}%
\newtheorem{definition}{Definition}%
\newtheorem{assumption}{Assumption}%

\usepackage{lastpage}

\usepackage[utf8]{inputenc} % allow utf-8 input
\usepackage[T1]{fontenc}    % use 8-bit T1 fonts
\usepackage{hyperref}       % hyperlinks
\usepackage{url}            % simple URL typesetting
\usepackage{booktabs}       % professional-quality tables
\usepackage{amsfonts}       % blackboard math symbols
\usepackage{nicefrac}       % compact symbols for 1/2, etc.
\usepackage{microtype}      % microtypography
\usepackage{lipsum}
\usepackage{fancyhdr}       % header
\usepackage{graphicx}       % graphics
\graphicspath{{media/}}     % organize your images and other figures under media/ folder
\usepackage{soul}
\title{Leveraging Optimal Transport for Distributed Two-Sample Testing: An Integrated Transportation Distance-based Framework
}

\author{ Yan Chen\thanks{All authors equally contributed to this work and are listed in the alphabetic order.}\\
  University of Science and
Technology of China \\
  Hefei\\
  \texttt{yanc@mail.ustc.edu.cn}\\
     \And
  Zhengqi Lin$^{*,}$\\
  Rutgers University \\
  New Brunswick\\
  \texttt{zhengqi.lin@rutgers.edu} \\
 } %% examples of more authors

\begin{document}

\maketitle

\begin{abstract}
This paper introduces a novel framework for distributed two-sample testing using the Integrated Transportation Distance (ITD), an extension of the Optimal Transport distance. The approach addresses the challenges of detecting distributional changes in decentralized learning or federated learning environments, where data privacy and heterogeneity are significant concerns. We provide theoretical foundations for the ITD, including convergence properties and asymptotic behavior. A permutation test procedure is proposed for practical implementation in distributed settings, allowing for efficient computation while preserving data privacy. The framework's performance is demonstrated through theoretical power analysis and extensive simulations, showing robust Type I error control and high power across various distributions and dimensions. The results indicate that ITD effectively aggregates information across distributed clients, detecting subtle distributional shifts that might be missed when examining individual clients. This work contributes to the growing field of distributed statistical inference, offering a powerful tool for two-sample testing in modern, decentralized data environments.
%  The Abstract paragraph should be indented 0.25 inch (1.5 picas) on both left and right-hand margins. Use 10~point type, with a vertical spacing of 11~points. The \textbf{Abstract} heading must be centered, bold, and in point size 12. Two line spaces precede the Abstract. The Abstract must be limited to one paragraph.
\end{abstract}

% keywords can be removed

\section{Introduction}
%\lipsum[2]
%\lipsum[3]Two sample test is a fundamental topic 

Let $\mu$ and $\nu$ be two underlying probability measures. Given two mutually independent samples $X_1, \ldots, X_m \sim \mu$ and $Y_1, \ldots, Y_n \sim \nu$. The goal of nonparametric two-sample testing is to determine whether the null hypothesis $H_0 : \mu = \nu$ holds, without imposing parametric assumptions on the distributions. A test statistic $T_{m,n}$ is a function of the sample data used to quantify the evidence against the null hypothesis in a statistical hypothesis test. A standard class of nonparametric two-sample tests rejects the null hypothesis $H_0: \mu = \nu$ if the test statistic $T_{m,n} = T_{m,n}(X_1, \ldots, X_m, Y_1, \ldots, Y_n)$ exceeds a critical value $c_{m,n}$, where $T_{m,n}$ is a scalar function of the mutually independent samples $X_1, \ldots, X_m \sim \mu$ and $Y_1, \ldots, Y_n \sim \nu$, and $c_{m,n}$ is chosen to achieve a desired significance level $\alpha \in (0,1)$. Such a sequence of tests is said to have asymptotic level $\alpha$ if $\limsup_{m,n \rightarrow \infty} \mathbb{P}(T_{m,n} > c_{m,n}) \leq \alpha$ whenever $\mu = \nu$, meaning that the probability of rejecting $H_0$ when it is true (Type I error) is at most $\alpha$ in the limit as the sample sizes $m$ and $n$ grow to infinity. Moreover, these tests are considered asymptotically consistent if $\lim_{m,n \rightarrow \infty} \mathbb{P}(T_{m,n} > c_{m,n}) = 1$ whenever $\mu \neq \nu$, implying that the probability of correctly rejecting $H_0$ when it is false (power) approaches 1 as the sample sizes increase. 
\begin{comment}

\st{In the subsequent analysis, it is assumed that both $m$ and $n$ tend to infinity, and their ratio $m/(m+n)$ converges to a constant $\tau \in (0,1)$.}

\end{comment}

\subsection{Two sample test for distributed devices}

We consider a distributed machine learning setting with numerous devices, assumed to be stateful and participating at each round, and a central server that coordinates training across the clients. The training data are decentralized and arrive over time. We want to explore the possibility of developing a two-sample test under such a distributed learning setting, which is widely applicable in areas such as social media optimization and medical diagnosis. Under this setting, it is generally assumed that the local devices' data are preserved privately and may originate from different distributions across devices. The data at each device $k = 1, \ldots, K$ are sampled from distributions $P^k$ or $Q^k$, respectively. We say that there is a difference at client $k$ if $P^k \neq Q^k$. Our goal is to test the following distributed two-sample testing:
$$H_0:\;P^k =Q^k \text{ for all } k, \quad \quad\text{versus}\quad \quad H_1:\;\text{There exists at least one }k \text{ such that }  P^k \neq Q^k.$$

Detecting differences in distributions in a distributed machine learning setting poses several challenges that can be largely categorized into two aspects: practical implementation and theoretical reliability. 

Regarding practical implementation, data privacy and security concerns may arise, as devices may be reluctant to share their raw data, making it difficult to directly compare distributions across clients. Communication constraints, such as limited bandwidth or network latency, and the heterogeneity across devices in terms of computational capabilities or hardware specifications, can hinder the transmission of large amounts of data or statistics to a central server for analysis. As the number of devices or clients increases, scalability becomes a concern, with computational complexity and communication overhead potentially becoming bottlenecks.

From a theoretical perspective, the data may exhibit heterogeneity across devices, violating the assumptions of many traditional two-sample tests. Detecting and adapting to the differences in distributions can be challenging, especially when the degree of differences is not uniform across devices. Furthermore, the two-sample testing procedure should be robust to various types of noise, outliers, or adversarial attacks that may be present in the distributed setting, while maintaining interpretability to understand the nature and sources of distribution differences.

The Optimal Transport distance or Wasserstein distance, rooted in optimal transport theory, has gained attention for its ability to capture geometric aspects of distributions \cite{Ramdas_Garcia_Cuturi_2015}. An extension of this distance, named the Integrated Transportation Distance (ITD), was introduced by \cite{lin2023integrated,lin2023} to establish a novel metric between stochastic (Markov) kernels. In the original work, the ITD was employed to construct approximations of Markov systems and assess the accuracy of these approximations. Building upon these findings, we recognize that the ITD has potential applications beyond Markov systems, particularly in two-sample testing under distributed or federated learning settings. 
ITD is particularly well-suited for two-sample testing in distributed settings for several reasons. It integrates statistical distances from individual devices without requiring centralized data access, allows for weighted contributions that account for heterogeneity among devices, and can provide accurate approximations even with partial client participation. Furthermore, ITD maintains interpretability and computational efficiency throughout the process. These characteristics collectively make ITD an ideal tool for detecting distributional changes in decentralized data environments.

\subsection{Relevant work}

Statistical distances have played a crucial role in two-sample testing, evolving from classical methods like the Kolmogorov-Smirnov test \cite{Massey_1951} and Wilcoxon-Mann-Whitney test \cite{mannhb,Wilcoxon1945IndividualCB} to more sophisticated approaches. Recent advancements include energy
distance \cite{szekely2004testing}, Maximum Mean Discrepancy (MMD) \cite{gretton12a}, ball divergence \cite{pan2018ball}, projection averaging-type Cramér–von Mises statistics \cite{kim2020robust}.  However, these methods are primarily designed for single-device settings and are unsuitable for distributed frameworks.

In the context of distributed hypothesis testing, \cite{chang2024statistical} developed a moving-subgroup algorithm  aimed at improving computational efficiency for a global testing, though this approach does not guarantee the data privacy. \cite{Cai2024FederatedNH} established the minimax separation rate for goodness-of-fit testing under privacy constraint.  Distributed hypothesis testing has found applications in areas such as sensor networks \cite{8293844}, multiple testing \cite{Liu2020IntegrativeHD}. To our knowledge, no existing distributed two-sample testing methods simultaneously address data privacy and heterogeneity—a critical challenge in real-world settings such as hospitals and schools.

\section{The Optimal Transport Distance}

The essential ingredient of our construction is the Optimal Transport (OT) distance or Wasserstein distance between measures. In here,
$\Xc$ is a Polish space, with the metric $d(\cdot,\cdot)$, and the associated Borel $\sigma$-field $\Bc(\Xc)$.
The symbol $\Pc(\Xc)$ denotes the space of probability measures on $\Bc(\Xc)$.
In the brief summary below, we follow \cite{villani2009optimal}. The reader is referred to this monograph, as well as to \cite{rachev1998mass}
for an extensive exposition and historical account.

On the space $\Pc(\Xc)$, the OT distance of order $p$, $W_{p}(\cdot,\cdot)$, is not a metric in the strict sense, because it might take the value $+\infty$. We restrict its domain to a subset of $\Pc(\mathcal{X})$  on which it only takes finite values.
    The \textit{measure space of order} $p$ is defined as
    \begin{equation*}
\Pc_{p}(\mathcal{X}) = \left\{\mu \in \Pc(\mathcal{X}) : \ \int_{\mathcal{X}} d\left(x_{0}, x\right)^{p} \;\mu(\D x)<+\infty\right\},
\end{equation*}
where $x_{0} \in \mathcal{X}$ is arbitrary. This space does not depend on the choice of the point $x_{0}$.

%The Wasserstein space $\Pc\emph{}_{p}(\mathcal{X})$ is therefore the space of probability measures which have a finite moment of order $p$.

\begin{definition} \label{Wass_d1}
 For two probability measures $\mu, \nu\in \Pc_p(\Xc)$, the OT distance of order $p \in[1, \infty)$ between $\mu$ and $\nu$ is defined by the formula
\begin{equation} \label{Wass}
\begin{aligned}
W_{p}(\mu, \nu) &=\left(\inf _{\pi \in \Pi(\mu, \nu)} \int_{\mathcal{X} \times \mathcal{X} } d(x, y)^{p} \;  \pi({\D x}, {\D y})\right)^{1 / p} \\
 &= \inf \left\{\left({\Eb} \big[d(X, Y)^{p}\big] \right)^{\frac{1}{p}}:\; \text {\rm law}(X)=\mu, \;\text {\rm law}(Y)=\nu\right\},
\end{aligned}
\end{equation}
where $\Pi(\mu, \nu)$, called the coupling set, is the set of all probability measures in $\Pc(\Xc\times\Xc)$ with the marginals $\mu$ and $\nu$.
\end{definition}
The measure $\pi^* \in \Pi(\mu, \nu)$ that realizes the infimum in Equation \eqref{Wass} is called the \emph{optimal coupling} or the \emph{optimal transport plan}. For each $p\in [1,\infty)$, the function $W_{p}(\cdot,\cdot)$ defines a metric on $\Pc_{p}(\mathcal{X})$. Furthermore, for all $\mu,\nu\in \Pc_{p}(\mathcal{X})$ the optimal coupling
realizing the infimum in \eqref{Wass} exists. From now on, $\Pc_{p}(\mathcal{X})$ will be always equipped with the distance $W_p(\cdot,\cdot)$.

The following classical result, known as the Kantorovich–Rubinstein duality, provides an alternative characterization of $W_1(\cdot,\cdot)$.
\begin{theorem}
For any $\mu, \nu$ in $\Pc_{1}(\mathcal{X})$,
\begin{equation*}\label{dual_wass}
W_{1}(\mu, \nu)=\sup _{\|\psi\|_{\text {\rm Lip}} \leq 1}\left\{\int_{\mathcal{X}} \psi(x)\; \mu({\D x})-\int_{\mathcal{X}} \psi(x) \; \nu({\D x})\right\},
\end{equation*}
where $\|\psi\|_{\text{\rm Lip}}$ denotes the minimal Lipschitz constant of the function $\psi:\Xc\to \Rb$.
\end{theorem}
In the discrete case, it follows from the linear programming duality.

For multidimensional state space, \cite{del2019central} establish asymptotic normality of $\sqrt{n}\left(W_2^2\left(\mu_m, \nu_n\right)-\mathbb{E}\left[W_2^2\left(\mu_m, \nu_n\right)\right]\right)$ by deriving an asymptotic linear representation using the Efron-Stein variance inequality. \cite[Thm 3.3]{del2019central} also provides variance bounds for the two-sample empirical transportation cost between two sets of i.i.d. random variables $X_1, \ldots, X_m$ with law $\mu$ and $Y_1, \ldots, Y_n$ with law $\nu$. If $\mu$ and $\nu$ have densities and finite fourth moments, then
\begin{equation*}
\mathrm{Var}(W_2^2(\mu_m, \nu_n)) \leq \frac{C_{\mu,\nu}}{m} + \frac{C_{\nu,\mu}}{n},
\end{equation*}
where $C_{\mu,\nu}$ is a constant depending on the moments of $\mu$ and $\nu$.

\section{The Integrated Transportation Distance Between Kernels} \label{s:kernel-distance}

We now introduce an essential concept for distributed two-sample testing: the Integrated Transportation Distance between kernels. This concept is the cornerstone of our approach. In previous literature, it was originally developed for risk evaluation in multi-stage Markov systems \cite{IISE1,IISE2,ZL1}. Its properties qualify it as a useful tool for dealing with statistical heterogeneity among many devices.
\

Suppose $\Xc$ and $\Yc$ are Polish spaces.  By the measure disintegration formula, every probability measure $\theta\in \Pc(\Xc\times\Yc)$ admits a disintegration
$\theta= {\lambda} \circledast Q$, where ${\lambda} \in \Pc(\Xc)$ is the marginal distribution on $\Xc$, and $Q:\Xc \to \Pc(\Yc)$ is a \emph{kernel} (a function
such that for each $B\in \Bc(\Yc)$ the mapping $x\mapsto Q(B|x)$ is Borel measurable):
\[
\theta(A \times B) = \int_A  Q(B|x)\;{\lambda}({\D x}), \forall \big(A\in \Bc(\Xc),B \in \Bc(\Yc)\big).
\]
Conversely, given a marginal ${\lambda} \in \Pc(\Xc)$  and a kernel $Q:\Xc \to \Pc(\Yc)$, the above formula defines a probability measure ${\lambda} \circledast Q$
on $\Xc \times \Yc$.
Its marginal on $\Yc$ is the \emph{mixture distribution} $\lambda \circ Q$ given by
\[
(\lambda \circ Q)(B) = \int_\Xc  Q(B|x)\;{\lambda}({\D x}), \quad \forall \, B \in \Bc(\Yc).
\]
To define a distance between kernels with the use of the Wasserstein metric in the space of probability measures,
we restrict the class of kernels under consideration. We use the same symbol $d(\cdot,\cdot)$ to denote the metrics on $\Xc$ and $\Yc$. 

\begin{definition}
\label{d:kernel-space} \cite[Defintion 3.1]{lin2023integrated}
The kernel space of order $p\in [1,\infty)$  is the set
\begin{align*}
\label{kernel-class}
\Qc_p(\Xc,\Yc) =& \Big\{ Q:\Xc \to \Pc_p(\Yc) : \forall \big(B\in \Bc(\Yc)\big) \; 
Q(B|\,\cdot\,) \text{ is Borel measurable}, \\
&\quad \exists (C>0)\,{  \forall(x\in \Xc)}  \int_\Yc d(y,y_0)^p\; Q({\D y}|x) \le C\big(1 + d(x,x_0)^p\big)\Big\}.
\end{align*}
\end{definition}
It is evident that the choice of the points $x_0\in \Xc$ and $y_0\in \Yc$ is irrelevant in this definition. The kernel space $\mathcal{Q}_p(\mathcal{X},\mathcal{Y})$ ensures that ITD takes finite values by imposing moment conditions on the kernels, which is crucial for two-sample testing as it guarantees the existence and well-definedness of the test statistic, allowing for meaningful comparisons between distributions and the development of rigorous statistical inference procedures in distributed learning settings.

\begin{definition}
\label{d:kernel-distance} \cite[Defintion 3.2]{lin2023integrated}
The integrated transportation distance of degree $p$ between two kernels $P$ and $Q$ in $\Qc_p(\Xc,\Yc)$ with fixed marginal ${\lambda}\in \Pc_p(\Xc)$ is defined as
\begin{equation*} \label{TS}
    \Wc_{p}^\lambda(P, Q)=\left(\int_{\Xc} \big[{W}_{p}(P(\,\cdot\, | x), Q(\,\cdot\, | x))\big]^p \;{\lambda}(\D x)\right)^{1/p} .
\end{equation*}
\end{definition}
%We observe that the distance is the norm of the function $x \mapsto {W}_{p}(Q(\cdot | x), \widetilde{Q}(\cdot | x))$ in the space
%$\Lc_p\big(\Xc,\Bc(\Xc),\lambda)$.
From now on, for a fixed marginal $\lambda\in \Pc_p(\Xc)$, we shall identify the kernels $P$ and $Q$
if ${W}_{p}(P(\,\cdot\, | x), Q(\,\cdot\, | x))=0$ for $\lambda$-almost all $x\in \Xc$. {Thus, we consider the
space $\Qc_p^\lambda(\Xc,\Yc)$ of equivalence classes of $\Qc_p(\Xc,\Yc)$, with a fixed marginal $\lambda\in \Pc_p(\Xc)$. }If $\lambda$ has finite support ${x_1, \ldots, x_K}$ with weights $w_1, \ldots, w_K$ respectively, where $\sum_{k=1}^Kw_k=1$. The integrated transportation distance of degree $p$ between $P$ and $Q$ can be written as
\begin{comment}
   \BLUE{Here,  we should unified ITD symbol (such as ${\mathrm{ITD}}_p$)??} 
\end{comment}
\begin{equation*} \label{TS-finite}
\Wc_{p}^\lambda(P, Q)=\left(\sum_{k=1}^K w_k\big[{W}_{p}(P(\cdot | x_k), Q(\cdot| x_k))\big]^p\right)^{1/p} .
\end{equation*}

\begin{theorem} \label{TS_T}\cite[Thm 3.3]{lin2023integrated} For any $p\in [1,\infty)$ and any $\lambda\in \Pc_p(\Xc)$, the function $\Wc_p^\lambda(\cdot,\cdot)$, defines a  metric on the space $\Qc_p^\lambda(\Xc,\Yc)$.
\end{theorem}

Theorem \ref{TS_T} establishes ITD as a proper metric on the space of kernels, which is fundamental for two-sample testing as it ensures that the distance between distributions is well-defined and satisfies essential properties like non-negativity, symmetry, and the triangle inequality.

ITD additionally serves as an upper bound for the distances between two mixture distributions and between two composition distributions.
For all $\lambda \in \Pc_p(\Xc)$ and all $P ,  Q \in \Qc_p^\lambda(\Xc,\Yc)$, we have
\begin{equation*}
\label{hierarchy}
\Wc_p^\lambda(P,Q) \geq {W}_p({\lambda} \circledast P,{\lambda} \circledast Q) \geq
{W}_p({\lambda} \circ P,{\lambda} \circ Q);
\end{equation*}
see \cite[Thm 3.5]{lin2023integrated}. This property is particularly beneficial for data privacy in distributed learning environments, as it enables the comparison of mixture distributions without requiring direct access to individual data points, instead relying on aggregated statistics that preserve the privacy of local datasets.

The Lipschitz continuity of the kernels plays an important role in our research.

\begin{assumption}
   A kernel $Q:\Xc \to \Pc_p(\Yc)$ is Lipschitz continuous if a constant $L_Q$ exists, such that
\begin{equation*}
    \label{Dobrushin}
W_p(Q(\,\cdot\, | x),Q(\,\cdot\, | x')) \le L_{Q}\, d(x,x'),\quad \forall\;x,x'\in \Xc.
\end{equation*}
\end{assumption}

\begin{comment}
\begin{lemma}
\label{f:muQ-cont}
If a kernel $Q:\Xc \to \Pc_p(\Yc)$ is Lipschitz continuous, then the mapping $\mu \mapsto \mu \circ Q$ is Lipschitz continuous on
$\Pc_p(\Xc)$ with the same modulus.
\end{lemma}
\end{comment}
The following theorem validates the use of ITD in distributed two-sample testing by ensuring that our empirical estimates converge to the true ITD as we increase both the number of sampled clients and the amount of data from each client. It justifies our approach of sampling a subset of clients to represent the overall distribution and collecting data from each selected client to estimate local distributions. %The proofs of the following theorems are attached in the appendix.

\begin{theorem}\label{thm:itd1}
    Let $p \in [1, \infty)$, $\mathcal{X}$ and $\mathcal{Y}$ be Polish spaces, $\lambda \in \mathcal{P}_p(\mathcal{X})$, and $P, Q \in \mathcal{Q}_p(\mathcal{X}, \mathcal{Y})$ be Lipschitz continuous kernels with Lipschitz constants $L_P$ and $L_{{Q}}$. Suppose $(\lambda_k)_{k \geq 1} \subset \mathcal{P}_p(\mathcal{X})$, $(P_m)_{m \geq 1} ,  (Q_n)_{n \geq 1} \subset \mathcal{Q}_p(\mathcal{X}, \mathcal{Y})$ are Lipschitz continuous sequences such that:
\begin{enumerate}
\item $\lambda_k \xrightarrow{} \lambda$ converges weakly in $\mathcal{P}_p(\mathcal{X})$,
\item $P_m \xrightarrow{} P$ converges weakly in $\mathcal{Q}_p(\mathcal{X}, \mathcal{Y})$,
\item $Q_n \xrightarrow{} Q$ converges weakly in $\mathcal{Q}_p(\mathcal{X}, \mathcal{Y})$.
\end{enumerate}
Then,

$$
\lim_{k,m,n \to \infty} \mathcal{W}_p^{\lambda_k}(P_m, Q_n) = \mathcal{W}_p^\lambda(P, Q).
$$

\end{theorem}

The following theorem extends the previous result by establishing the asymptotic behavior of the empirical ITD estimator. It provides a Central Limit Theorem for the ITD, demonstrating that the difference between the empirical ITD (based on a finite sample of clients) and the true ITD converges in distribution to a Gaussian process as the number of sampled clients increases. This result is important for constructing confidence intervals and performing hypothesis tests using the ITD in distributed learning settings, allowing for statistical inference about distributional differences across the network of clients.

\begin{theorem}\label{thm:process}
    Let $\mathcal{X}$ and $\mathcal{Y}$ be Polish spaces, $p \geq 2$, and $\lambda \in \mathcal{P}_p(\mathcal{X})$. Suppose $P, Q \in \mathcal{Q}_p(\mathcal{X}, \mathcal{Y})$ are Lipschitz continuous kernels with Lipschitz constants $L_P$ and $L_{{Q}}$, respectively and the Wasserstein distance $W_p(\cdot, \cdot)$ between any two measures in the space $\mathcal{Q}_p(\mathcal{X}, \mathcal{Y})$ is bounded by a constant $R$. Let $x_1, \ldots, x_K$ be an independent and identically distributed (i.i.d.) sample from $\lambda$, and let $\lambda_{K}$ be the $K$th empirical measure. Define the empirical integrated transportation distance as
$$
\Wc_p^{p, \lambda_{K}}(P, {Q}) = \frac{1}{K} \sum_{k=1}^K \big[W_p(P(\cdot|x_k), {Q}(\cdot|x_k))\big]^p.
$$
Then, as $K \rightarrow \infty$, we have
$$
\sqrt{K}\big[\Wc_p^{p, \lambda_{K}}(P, {Q}) - \Wc_p^{p, \lambda}(P, {Q})\big] \xrightarrow{\mathcal{D}} B,
$$
where $B$ is a zero-mean Gaussian process with variance function
$$
\mathbf{var}(B) = \int_\mathcal{X} \big[W_p^p(P(\cdot|x), {Q}(\cdot|x)) - \Wc_p^{p, \lambda}(P, {Q})\big]^2 \lambda(\D x).
$$
\end{theorem}

\section{Two sample test for distributed learning}

\subsection{Motivation}

In the context of distributed learning, two-sample tests can be employed to detect changes in the underlying data distribution across different clients. In this scenario, the central server needs to regularly monitor and determine whether the data from each client comes from the same distribution. However, the challenges of detecting these changes are amplified due to the distributed nature of the data and the potential for local shifts to occur abruptly for some clients while not occurring for others. When different clients experience data shifts at varying times, efficient and effective two-sample tests in the distributed learning setting become necessary to address these challenges, and the development of mathematical tools and evaluation metrics is necessary to answer this question.

Therefore, the lack of two-sample tests specifically designed for distributed learning settings has motivated us to explore the Integrated Transportation Distance as a potentially ideal tool. ITD is particularly well-suited for this context, as it takes into account the statistical heterogeneity among devices while preserving individual data privacy. Moreover, recent theoretical advancements, such as the establishment of theorems concerning sample complexity and the central limit theorem for the Wasserstein distance and its variations, have inspired us to investigate and develop a deeper understanding of the statistical bounds for ITD. Building upon these insights, we aim to develop robust mechanisms and methodologies for conducting two-sample tests across distributed devices, thereby addressing the unique challenges and opportunities presented by this emerging field of research. 
\begin{comment}
 In a more formal framework, we propose using the p-th order of the Integrated Transportation Distance, denoted as $ \Wc_{p}^{p,\lambda}(Q, \widetilde{Q})$, as the test statistic and the basis for determining a threshold value. Specifically, if the p-ITD exceeds a certain threshold $c$, i.e., $\Wc_{p}^{p,\lambda}(Q, \widetilde{Q}) > c$, we conclude that there is a significant shift from the kernel $Q$ to the kernel $\widetilde{Q}$. This approach allows for a rigorous and quantifiable assessment of the dissimilarity between the two kernels, enabling us to make informed decisions regarding the presence of distributional shifts in the context of distributed learning.
In cases where the total number of devices is too large to count explicitly, our approach is to select a sample of $K$ devices.
For $k = 1, \ldots, K$, let $\{\mathbf{x}^k_i, i=1, \ldots, m_k\}$ and $\{\mathbf{y}^k_i, i=1, \ldots, n_k\}$ be two mutually independent random samples drawn from $Q(\cdot|z^k)$ and $\widetilde{Q}(\cdot|z^k)$, with sample sizes $m_k$ and $n_k$, respectively. We then compute an estimate of $\Wc_{p}^{p,\lambda}(Q, \widetilde{Q})$, which we denote as $\widehat{\Wc}_{p}^{p,\lambda_K}(Q_m, \widetilde{Q}_n)$ or the empirical integrated transportation distance.
   
\end{comment}

In a more formal framework, we propose using the $p$-th order of the Integrated Transportation Distance of the kernel $P$ and the kernel $Q$, denoted as $ \Wc_{p}^{p,\lambda}(P, Q)$, as the test statistic and the basis for determining a threshold value. Specifically, if the $p$-ITD exceeds a certain threshold $c$, i.e., $\Wc_{p}^{p,\lambda}(P, Q) > c$, we conclude that there is a significant shift from the kernel $P$ to the kernel $Q$. This approach allows for a rigorous and quantifiable assessment of the dissimilarity between the two kernels, enabling us to make informed decisions regarding the presence of distributional shifts in the context of distributed learning.

In cases where the total number of devices is too large to count explicitly, our approach is to select a sample of $K$ devices.
For $k = 1, \ldots, K$, let $\{\mathbf{x}^k_i, i=1, \ldots, m_k\}$ and $\{\mathbf{y}^k_j, j=1, \ldots, n_k\}$ be two mutually independent random samples drawn from $P^k$ and $Q^k$, with sample sizes $m_k$ and $n_k$, respectively. For convenience, we
assume $m_1=\cdots=m_K=m$ and $n_1=\cdots=n_K=n$.  We then compute an estimate of $\Wc_{p}^{p,\lambda}(P, Q)$, which we denote as ${\Wc}_{p}^{p,\lambda_K}(P_m, Q_n)$ or the empirical integrated transportation distance, where $P_m$ and $Q_n$ represent the empirical distributions from kernels $P$ and $Q$.

\subsection{Empirical Integrated Transportation Distance }

We set $p=2$ for practical situations, which is attractive for theoretical analysis. Denote $P_m^k$ and $Q_n^k$ as the empirical distributions as follows
\[
 P^k_m:=\frac{1}{m}\sum_{i=1}^{m}\delta_{{\bf x}^k
_{i}}, \quad  Q^k_n:=\frac{1}{n} \sum_{j=1}^{n} \delta_{{\bf y}^k_{j}}.
\]
Then the second-order empirical ITD is
\begin{equation}\label{eq:ITD}
  \begin{aligned}
  &\widehat{\mathrm{ITD}}^2=\mathcal{W}_2^{2,\lambda_K}(P_m, {Q}_n)=  \sum_{k=1}^K w_k\big[{W}_{2}(P^k_m, Q^k_n)^2\big],
    \end{aligned}
\end{equation}
and the expectation of empirical ITD is
\begin{equation*}
  \begin{aligned}
 &\mathbf{E}(\widehat{\mathrm{ITD}}^2)= \sum_{k=1}^K w_k\mathbf{E}\big[{W}_{2}(P^k_m, Q^k_n)^2\big],
  \end{aligned}
\end{equation*}
where the weight $\omega_k$ can be determined in various ways, depending on the specific scenario. In this paper, we consider the weight $\omega_k$ to be related to the ratio of sample sizes, with $\omega_k=\frac{m_k}{2M}+\frac{n_k}{2N}$, where $M=\sum_{k=1}^K m_k$ and $N=\sum_{k=1}^K n_k$ represent the total sample sizes.  In the case of equal sample sizes, the weight simplifies to $\omega_k=\frac{1}{K}$, for all $k=1,\ldots,K$.

Next, leveraging McDiarmid’s inequality, we provide a concentration result for the large deviation bound of the empirical ITD. %(Lemma \ref{lm:Mc})
\begin{theorem}\label{thm:concen}
For each client, assume $P^k$ and $Q^k$ have densities and second moments, and let $\omega_k=\frac{1}{K},k=1,\ldots,K$. Suppose $D_x=\max_{k=1,\ldots,K }\{\|{\bf x}^k\|^2,{\bf x}^k\sim P^k\}$ and $D_y=\max_{k=1,\ldots,K }\{\|{\bf y}^k\|^2,{\bf y}^k\sim Q^k\}$ are finite, for all $t\in \mathbb{R}$, then
%\small
 \begin{equation*}
 \begin{aligned}
      &\mathbf{P}\left(\widehat{\mathrm{ITD}}^2-\mathbf{E}\left(\widehat{\mathrm{ITD}}^2\right) >t\right)\leq \exp\left(-\frac{Kmnt^2}{{2(m+n)(D_x+D_y)^2}}\right).
  \end{aligned}
 \end{equation*}  
\end{theorem}

\subsection{Critical value and permutation test}

In this distributed two-sample testing procedure, local devices or clients are responsible for computing Wasserstein distances between their sample sets, both for the original data and for multiple permutations. They perform these calculations internally, preserving data privacy. The central (global) device, on the other hand, coordinates the overall process. It selects the participating devices, assigns weights, aggregates the local Wasserstein distances into the Integrated Transportation Distance statistic, generates permuted ITD values, and performs the final statistical inference. This division of labor allows for efficient use of distributed computational resources while maintaining data confidentiality, as only summary statistics are shared with the central device.

Therefore, we propose the permutation test procedure as follows: 

\begin{enumerate}
    \item The global device randomly selects $K$ clients from the total number of devices. For each device $k = 1, \ldots, K$, we assign a weight $w_k$ such that $\sum_{k=1}^K w_k= 1$.
\begin{comment}

    \item For each device $k$, let ${\mathbf{x}^k_i, i=1, \ldots, m_k}$ and ${\mathbf{y}^k_i, i=1, \ldots, n_k}$ be samples with sizes $m_k$ and $n_k$, respectively. The empirical distributions obtained from these samples are denoted as $Q_{m_k}$ and $\widetilde{Q}_{n_k}$. We then compute $\widehat{T}^{k} = W_2(Q_{m_k},\widetilde{Q}_{n_k})$.
\end{comment}
     \item For each local device $k$, let ${\mathbf{x}^k_i, i=1, \ldots, m}$ and ${\mathbf{y}^k_j, j=1, \ldots, n}$ be samples with sizes $m$ and $n$, respectively.  We then compute $W_2(P^k_{m},Q^k_{n})$. The central device will compute the test statistics $ \widehat{\mathrm{ITD}}^2$ using \eqref{eq:ITD}. %The empirical distributions obtained from these samples are denoted as $P^k_{m}$ and $Q^k_{n}$.
    
    \item Let $\mathcal{Z}_{m+n}^k=\left\{\mathbf{z}^k_1, \mathbf{z}^k_2, \ldots, \mathbf{z}^k_{m+n}\right\}=\left\{\mathbf{x}^k_1, \ldots, \mathbf{x}^k_m, \mathbf{y}^k_1, \ldots, \mathbf{y}^k_n\right\}$ denote the pooled samples for device $k = 1, \ldots, K$. Each device randomly permutes the pooled samples to obtain $\left\{\mathbf{z}_1^{\pi,k}, \mathbf{z}_2^{\pi,k}, \ldots, \mathbf{z}_{m+n}^{\pi,k}\right\}$.
    
    \item For each device $k$, select the first $m$ observations from the pooled samples as $\left\{\mathbf{x}_1^{\pi,k}, \ldots, \mathbf{x}_m^{\pi,k}\right\}$, and the rest observations as $\left\{\mathbf{y}_1^{\pi,k}, \ldots, \mathbf{y}_n^{\pi,k}\right\}$.
    
    \item For each device $k$, based on two randomly permuted samples $\left\{\mathbf{x}_1^{\pi,k}, \ldots, \mathbf{x}_m^{\pi,k}\right\}$ and $\left\{\mathbf{y}_1^{\pi,k}, \ldots, \mathbf{y}_n^{\pi,k}\right\}$. Denote the empirical permutation distribution functions for each client  as $P^{\pi,k}_m$ and ${Q}^{\pi,k}_n$. One can calculate the test statistic to obtain $\widehat{T}^{\pi,k}=W^2_2(P^{\pi,k}_{m},Q^{\pi,k}_{n})$.
    \item For each device $k$, repeat steps 3 to 5 for $B_k$ times to obtain $\widehat{T}_{b_k}^{\pi,k}, b_k= 1, \ldots, B_k$ and return the results to the central device.
    \item The central device randomly select one permuted test statistics $\widehat{T}_{b_k}^{\pi,k}$ for each device $k$ and calculate the integrated test statistics as $\widehat{\mathrm{ITD}}^{2,\pi} = \sum_{k=1}^K w_k\widehat{T}_{b_k}^{\pi,k}$. Repeat this process $B$ times to obtain $B$ different values of permuted integrated statistics, denoted as $\widehat{\mathrm{ITD}}_b^{2,\pi}$ for $b = 1, \ldots, B$. 
    
    \item The critical value computed by the global device, is
    \begin{small}
       \begin{equation*}
        c_{1-\alpha}^{m,n,K}:=\min\left\{ z\in \mathbb{R}:\frac{1}{B}\sum_{b=1}^B \mathbbm{1} \left(\widehat{\mathrm{ITD}}_b^{2,\pi}< z\right) \geq 1-\alpha\right\}.
    \end{equation*}   
    \end{small} 
   % The associated $p$-value is given by 
   % $$\frac{1}{B}\sum_{b=1}^B \mathbbm{1}\left(\widehat{\mathrm{ITD}}_b^{2,\pi} \geq  \widehat{\mathrm{ITD}}^2 \right),$$

    %where $I(\cdot)$ is an indicator function. 

    \item  Reject the null hypothesis if the test statistic satisfies $\widehat{\mathrm{ITD}}^2\geq c_{1-\alpha}^{m,n,K}$; otherwise, fail to reject the null hypothesis.%Reject the null hypothesis if the $p$-value is smaller than the given significance level $\alpha$.

\end{enumerate}

Since the integrated test statistic $ \widehat{\mathrm{ITD}}^2$ is the weighted average of test statistics for each selected device, if we obtain $B_k$ results for each device $k$, we can find a total of $\prod_{k=1}^K B_k$ permuted integrated test statistics.

\textbf{Time Complexity Analysis:} Each of the $k$th local devices performs $O((B_k + 1) * o(m,n))$ operations, where $o(m,n)$ is the cost of calculating the Wasserstein distance between two samples of sizes $m$ or $n$. In multivariate settings, this cost is of order $O(\max(m,n)^3)$. The global server's operations, including client selection, result aggregation, and critical value computation, have a complexity of $O(K*B + B*log(B))$ or $O(K*B)$ if $B$ is significantly smaller than $K$. Thus, the overall time complexity is $O(K * (B_k + 1) * o(m,n) + K * B + B * \log(B))$. In scenarios where the sample sizes $m$ or $n$ are large, the computational burden on the global server becomes relatively insignificant compared to the local computations, highlighting the efficiency of this distributed approach in handling large-scale datasets while maintaining data privacy.

\subsection{Theoretical properties}
 According to \cite{lehmann1986testing} and \cite{hemerik2018exact}, the permutation test provides finite-sample Type I error control when the null hypothesis is true. Since the observations on each distributed machine are exchangeable under $H_0$, then the Type I error of ITD can be uniformly bounded by the pre-specified level $\alpha$. Next, under the alternative hypothesis, we demonstrate that ITD maintains high power within the framework of federated learning, regardless of the number of clients. This is because, under the alternative hypothesis $H_1$, we can always guarantee that $\widehat{\mathrm{ITD}}^2\geq\delta$, where $\delta$ is a constant. Theorem \ref{thm:alter} is the critical theorem for power analysis of practical validation, its result is derived from the conclusions of Theorem \ref{thm:concen}.
\begin{theorem}\label{thm:alter}
 Under the alternative hypothesis $H_1$, let $\omega_k=\frac{1}{K},k=1,\ldots,K$.  If $P^k$ and $Q^k$  possess densities and have finite fourth moments, whenever $K$ approaches infinity or remains fixed, we have %Assume that the probability densities of $P^k$ and $Q^k$ are H$\ddot{o}$lder continuous on a bounded connected domain $\Omega_k$ with Lipschitz boundary. Let $m=\min\{m_1,\ldots,m_K\}$ and $n=\min\{n_1,\ldots,n_K\}$.   and suppose $K=o\{\log n (m\land n)\}$,
\begin{equation*}
 \lim_{m,n\rightarrow \infty}\mathbf{P}\left(\widehat{\mathrm{ITD}}^2\geq c_{1-\alpha}^{m,n,K} \mid H_1\right)= 1. 
\end{equation*}
\end{theorem}

\section{Numerical Illustrations}

All experiments were conducted on a MacBook Pro equipped with an Apple M2 Pro processor and 16 GB of memory. For simplicity, we assume that every client has equal weight.

\subsection{Simulation}

In this section, we present the numerical performance of the proposed ITD test. We use the permutation test to obtain an empirical distribution and derive the $p$-value of the ITD statistic. The comparisons are based on Type-I error and power. Each experiment involves at most 10 clients ($K = 10$). The weights of all clients are the same. For each client, the sample size is fixed at 250 ($m, n = 250$), and the permutated test statistics are run and computed 100 times ($B_k = 100$). After the results of the permutated test statistics from all clients are sent to the central server, we obtain $B$ instances of the permutated integrated statistics ($B = 1000$). For simplicity, we assume that every client has equal weight. 

For the Type I error and power evaluation, we consider the following three types of distributions. The mean vectors are \((\mu_{x_1}, \ldots, \mu_{x_K})\) and \((\mu_{y_1}, \ldots, \mu_{y_K})\), and the covariance matrices are \(\text{diag}(\sigma_{x_k})_d\) and \(\text{diag}(\sigma_{y_k})_d\) for \(k = 1, \ldots, K\), where \(\text{diag}(\cdot)_d\) denotes diagonal matrices with dimension \(d\). The distributions considered are:
\begin{enumerate}
    \item Distribution 1: Multivariate normal distribution.
    \item Distribution 2: Multivariate log-normal distribution. A multivariate log-normal distribution is defined such that \(\log (X_k) \sim N\left(\mu_{x_k}, \text{diag}(\sigma_{x_k})_d\right)\) and \(\log (Y_k) \sim N\left(\mu_{y_k}, \text{diag}(\sigma_{y_k})_d\right)\) for \(k = 1, \ldots, K\).
    \item Distribution 3: Multivariate t distribution with 5 degrees of freedom.

\end{enumerate}

Additionally, we examine Model A for three types of distributions and report the Type I error. We also evaluate Models B through D for these distributions by varying the location parameter and/or scale parameter for the power evaluation.

\begin{enumerate}[label=\Alph*.]

    \item \(\left\{ \mu_{x_k}=\mu_{y_k} = u_k, \sigma_{x_k} = \sigma_{y_k} = r_k\right\}\) - same mean, and same variance.
    \item \(\left\{ \mu_{x_k}=\mu_{y_k} = u_k, \sigma_{x_k} = r_k, \sigma_{y_k} = r_k + \varepsilon_k \right\}\) - same mean, and different variance.
    \item \(\left\{ \mu_{x_k} = u_k, \mu_{y_k} = u_k + \epsilon_k, \sigma_{x_k} = \sigma_{y_k} = r_k \right\}\) - different mean, and same variance.
    \item \(\left\{\mu_{x_k} = u_k, \mu_{y_k} = u_k + \epsilon_k, \sigma_{x_k} = r_k, \sigma_{y_k} = r_k + \varepsilon_k \right\}\) - different mean, and different variance.
\end{enumerate}

Where \( u_k \) is a \( d \)-dimensional vector randomly generated from a uniform distribution between -1 and 1, and \( r_k \) is a \( d \)-dimensional vector randomly generated from a uniform distribution between 0.8 and 1.2. The terms \(\epsilon\) and \(\varepsilon\) are random errors that are small compared to the mean vector \( u \) and the variance vector \( \sigma \). In this case, both \(\epsilon\) and \(\varepsilon\) are set to be random normal errors with a mean 0 and a standard deviation of 0.25.

The Type I errors and power are reported in the tables for all three distributions, with different numbers of clients and dimensions. The permutation tests were repeated 200 times and the test significance level is set to be 0.05. Table 1 shows that ITD generally controls Type I error well, with most values close to the nominal 0.05 level across different numbers of clients ($K$) and dimensions (d) for all three distributions (Normal, Log-normal, and t-distribution). Tables 2, 3, and 4 demonstrate that ITD has high power in detecting differences between distributions and in detecting changes in both location (mean) and scale (variance) parameters. In addition, there's a general trend of improved power as the number of clients ($K$) increases, especially for the Log-normal distribution, supporting the idea that ITD can leverage information from multiple clients effectively.

\begin{table}[h]
 \caption{Performance of Type I error in Model A}
  \centering
  \begin{tabular}{ccccc}
    \toprule
    
    $K$     & $d$     & Normal & Log-normal & $t$-distribution \\
    \midrule
    1 & 2 & 0.060   &0.055 & 0.075 \\
     1 & 5 & 0.055   &0.075 &0.060  \\
     1 & 10 & 0.065   &0.070 &0.090  \\
     2 & 2 & 0.045   &0.060 &0.060  \\
     2 & 5 & 0.045   &0.050 &0.060  \\
     2 & 10 & 0.060   &0.035 &0.035  \\
     5 & 2 & 0.030   &0.065 &0.035 \\
     5 & 5 & 0.045  &0.065 &0.045 \\
     5 & 10 & 0.040   &0.060 &0.045  \\
     10 & 2 & 0.060   &0.050 &0.030  \\
     10 & 5 & 0.030   &0.040 &0.080  \\
     10 & 10 & 0.050   &0.035 &0.030 \\
   
    \bottomrule
  \end{tabular}
  \label{tab:table2}
\end{table}

\begin{table}[h]
 \caption{Performance of Power in Model B}
  \centering
  \begin{tabular}{ccccc}
    \toprule
    
    $K$     & $d$     & Normal & Log-normal & $t$-distribution \\
    \midrule
    1 & 2 & 0.970  &0.545 & 0.740 \\
     1 & 5 & 1   &0.875 &0.920  \\
     1 & 10 & 1  &0.985 &0.995  \\
     2 & 2 & 1 &0.650 &0.990  \\
     2 & 5 & 1 &0.885&0.985  \\
     2 & 10 & 1 &0.980  &1  \\
     5 & 2 & 1 &0.785 &0.985 \\
     5 & 5 & 1 &0.965 &1 \\
     5 & 10 & 1 &0.995  &1  \\
     10 & 2 & 1 &0.745 &0.995  \\
     10 & 5 & 1 &0.985 &1  \\

     10 & 10 & 1 &0.990 &1  \\
   
    \bottomrule
  \end{tabular}
  \label{tab:table2}
\end{table}

\begin{table}[h]
 \caption{Performance of Power in Model C}
  \centering
  \begin{tabular}{ccccc}
    \toprule
    
    $K$     & $d$     & Normal & Log-normal & $t$-distribution \\
    \midrule
    1 & 2 & 1  &0.710 & 0.995 \\
     1 & 5 & 1   &0.940 &1  \\
     1 & 10 & 1  &0.970 &1  \\
     2 & 2 & 1 &0.755 &1  \\
     2 & 5 & 1 &0.965&1  \\
     2 & 10 & 1 &0.995  &1  \\
     5 & 2 & 1 &0.835 &1 \\
     5 & 5 & 1 &0.975 &1 \\
     5 & 10 & 1 &0.990  &1  \\
     10 & 2 & 1 &0.905 &1  \\
     10 & 5 & 1 &0.990 &1  \\

     10 & 10 & 1 &0.990 &1  \\
   
    \bottomrule
  \end{tabular}
  \label{tab:table2}
\end{table}

\begin{table}[h]
 \caption{Performance of Power in Model D}
  \centering
  \begin{tabular}{ccccc}
    \toprule
    
    $K$     & $d$     & Normal & Log-normal & $t$-distribution \\
    \midrule
    1 & 2 & 1  &0.525 & 0.995 \\
     1 & 5 & 1   &0.975 &0.995  \\
     1 & 10 & 1  &0.870 &1  \\
     2 & 2 & 1 &0.585 &1  \\
     2 & 5 & 1 &0.980 &1  \\
     2 & 10 & 1 &0.965  &1  \\
     5 & 2 & 1 &0.885 &0.995 \\
     5 & 5 & 1 &0.950  &1 \\
     5 & 10 & 1 &0.995  &1  \\
     10 & 2 & 1 &0.955 &1  \\
     10 & 5 & 1 &0.985 &1  \\

     10 & 10 & 1 &1 &1  \\
   
    \bottomrule
  \end{tabular}
  \label{tab:table2}
\end{table}

\subsection{ITD for Concept Drift Detection in Distributed Learning}

Data distribution shifts over time, known as concept drift, can significantly impact model performance. In distributed learning environments, identifying such shifts presents unique challenges beyond the typical issues of statistical heterogeneity, communication costs, and privacy concerns. The decentralized nature of distributed learning means that drifts can occur at different times across various participants, making it difficult for a single global model to maintain optimal performance for all nodes.
The complexities of detecting concept drift in distributed learning settings are extensive. Local drifts may occur suddenly for some participants while remaining stable for others. This raises questions about how significant these shifts need to be to warrant updates to an existing production model. Developing robust mathematical tools and evaluation metrics is crucial to address these concerns effectively.

Integrated Transportation Distance offers a promising solution for detecting concept drift in distributed systems. ITD aggregates statistical distances across participants without direct access to local data, balancing statistical heterogeneity and participant importance. This approach enables detection of subtle distribution shifts that might be missed when examining individual clients, potentially improving statistical power.
ITD's interpretability allows easy identification of nodes experiencing larger shifts. Theoretical work suggests accurate ITD approximations are possible without full participation. Compared to other concept drift detection methods in distributed learning, ITD offers advantages in calculation ease and interpretation. We demonstrate ITD's capability to identify local drifts using the MNIST dataset.

Our experiment comprises 10 clients, each representing images of a unique digit from 0 to 9. For each client $k$, we sample two sets of images and compute the test statistic between them: a. $n_k$ = 100 images of the client's specific digit. b. $n_k$ = 100 images with a probability $\varepsilon$ from the client's digit and probability $1-\varepsilon$ from all digits. This sampling method introduces gradual concept drift by incorporating images from other digits. The Integrated Transportation Distance then aggregates information from all clients to detect global drift. We conduct a permutation test and repeat the experiment 200 times to estimate the power of the test. This experimental design allows us to evaluate the effectiveness of ITD in detecting concept drift within a distributed learning framework, while also providing insights into its performance across different drift scenarios and client configurations. In the table below, we observed that through aggregating information across all clients, an integrated test statistic can detect smaller overall shifts in distribution that might not be apparent when looking at individual clients. This leads to improved statistical power, especially when individual clients have limited data.

\begin{table}[H]
 \caption{Performance of Power }
  \centering
  \begin{tabular}{cccccccccccc}
    \toprule
    
    $\varepsilon$ & 0 & 1 & 2 & 3&  4 & 5& 6 & 7 & 8& 9 & ITD \\
    \midrule
    0.8 & 0.740 & 0.955& 0.300  & 0.445& 0.420 & 0.185& 0.620 & 0.550 & 0.205& 0.315 &1 \\

    0.9 & 0.400  & 0.945& 0.045& 0.085& 0.080 & 0.025& 0.215& 0.205& 0.030 & 0.115 &0.635 \\

    \bottomrule
  \end{tabular}
  \label{tab:table2}
\end{table}

\subsection{Approximation Quality of Approximated ITD with Varying Client Participation: A PM2.5 Case Study}

To further explore the performance and characteristics of the Integrated Transportation Distance in the context of distributed learning, we conduct a numerical study using the PM2.5 data\cite{liang2021modeling}. This dataset provides a rich source of spatio-temporal information on air pollution across China, making it ideal for investigating the properties of ITD in a real-world distributed setting. In this study, we focus on two key aspects: 1. The effect of varying $K$, the total number of clients (monitoring stations) selected, on the approximated ITD value. 2. The relationship between the approximated ITD and the true ITD as $K$ changes.

We analyze city-level daily PM2.5 concentration data from two time periods: September 1 to December 31 in both 2015 and 2016, encompassing 313 cities across China. Our study involves two main steps: 1. We conduct a permutation test to evaluate the Integrated Transportation Distance between these two time periods, varying the number of participating cities ($K$). The ITD value computed using all 313 cities is 458.36, and the permutation test statistic at the 95\% significance level is 191.27, indicating a significant difference in air quality between the two time periods. As shown in Table \ref{tab:pm2.5-k}, we observe a general trend of improved statistical power as the number of participating cities increases. 2. We randomly select $K$ cities from the total pool and compute an approximated ITD between the two time periods for these selected cities. We then compare these results to the ITD and permutation test outcomes calculated using all 313 cities. By repeating this process for different values of $K$ and multiple random selections, we observe that the distribution of the approximated ITD converges towards a normal distribution and centers around the value of true ITD, as illustrated in Figures \ref{fig:K50} to \ref{fig:K150} .

\begin{table}[h]
 \caption{Performance of Power with Varying $K$}
  \centering
  \begin{tabular}{ll}
    \toprule
    
    $K$     & power\\
    \midrule
    10 & 0.9 \\
     50 & 1   \\
     100 & 1 \\

    \bottomrule
  \end{tabular}
 \label{tab:pm2.5-k}
\end{table}

\begin{figure}[h]
\centering
\includegraphics[width=0.5\textwidth]{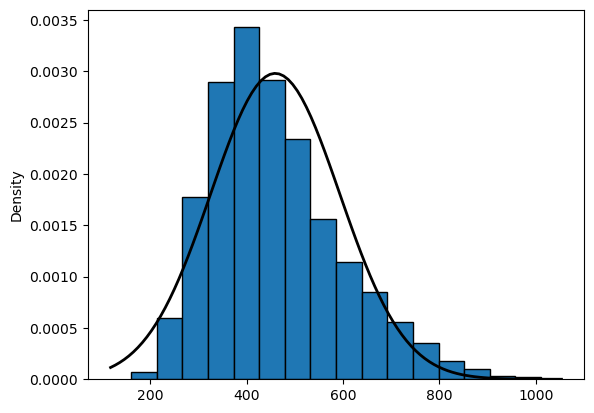}
\caption{Distribution of the Approximate ITD with $K =50$}
\label{fig:K50}
\end{figure}

\begin{figure}[h]
\centering
\includegraphics[width=0.5\textwidth]{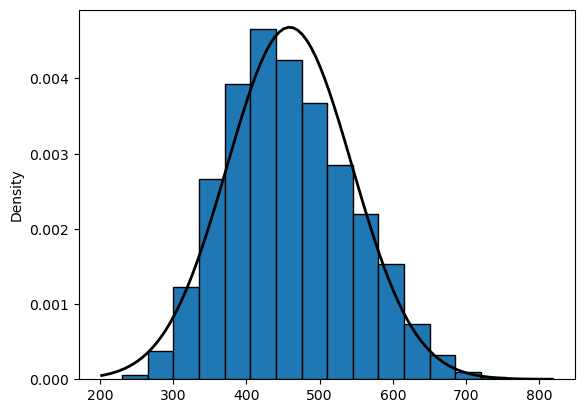}
\caption{Distribution of the Approximate ITD with $K =100$}
\label{fig:K100}
\end{figure}

\begin{figure}[h]
\centering
\includegraphics[width=0.5\textwidth]{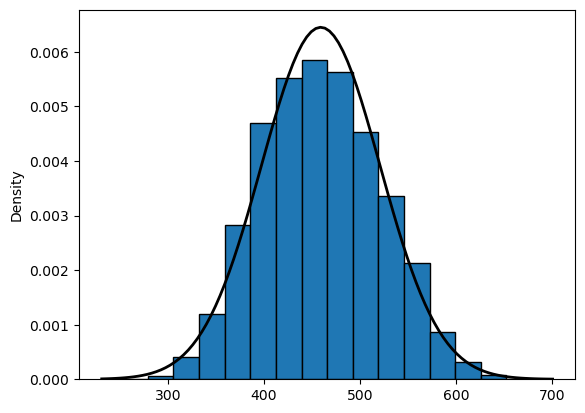}
\caption{Distribution of the Approximate ITD with $K =150$}
\label{fig:K150}
\end{figure}

\section{Conclusion}
This paper presents a novel framework for distributed two-sample testing using the Integrated Transportation Distance. Our approach addresses the unique challenges posed by distributed learning and federated learning environments, including data privacy concerns, statistical heterogeneity, and the need for efficient communication. The Integrated Transportation Distance proves to be a powerful tool for detecting distributional changes across multiple clients without requiring centralized access to raw data.

We have provided a solid theoretical foundation for the use of ITD in distributed settings, including convergence properties and asymptotic behavior. Our permutation test procedure offers a practical implementation strategy that balances computational efficiency with statistical rigor. Through extensive simulations, we demonstrated that our method maintains robust Type I error control and high statistical power across various distributions and dimensions.

Our work contributes significantly to the field of distributed statistical inference, offering researchers and practitioners a powerful tool for two-sample testing in modern, decentralized data environments. The ITD-based framework provides a promising approach to address the growing need for privacy-preserving, statistically sound methods in distributed learning settings.

Future research directions for this work include extending the ITD framework to multi-sample testing scenarios, allowing for comparisons across multiple time periods or data sources in distributed settings. We also aim to extend the ITD framework to incorporate entropic optimal transport, which could potentially improve computational efficiency and provide smoother optimal transport plans in high-dimensional settings. Additionally, adapting the ITD framework to handle data in non-Euclidean spaces, such as graphs or manifolds, would significantly broaden its applicability to diverse data types and problem domains. 

%\section*{Acknowledgments}
%This was was supported in part by......

%Bibliography

\newpage

\appendix
\section{Supplementary Theorems and Remarks on Optimal Transport Distance}

\begin{remark}
\label{r:discrete}
{\rm
Equation \eqref{Wass} has a convenient linear programming representation for discrete measures.
Let $\mu$ and $\nu$ be discrete measures in $\Pc(\mathcal{X})$, supported at positions $\{x^{(i)}\}_{i=1}^{m}$ and $\{y^{(j)}\}_{j=1}^{n}$ with normalized (totaling 1) positive weight vectors $w_{x}$ and $w_{y}$:
\[
\hat\mu=\sum_{i=1}^{m} w_{x}^{(i)} \delta_{x^{(i)}}, \quad \hat\nu=\sum_{j=1}^{n} w_{y}^{(j)} \delta_{y^{(j)}}.
\]
For $p \geq 1$, let $D \in {R}_{+}^{m \times n}$ be the distance matrix defined as $D_{i j}=d\big(x^{(i)},y^{(j)}\big)^{p}$. Then the $p$th power of the $p$-Wasserstein distance between the measures $\mu$ and $\nu$ is the optimal value of the following transportation problem:
\begin{equation*} \label{12}
\min _{\pi \in {R}_{+}^{m \times n}}\  \sum_{i=1}^m\sum_{j=1}^n  D_{i j} \pi_{i j} \quad
\text {s.t.} \quad  \pi^\top \1_{m}=w_{x}, \quad  \pi \1_{n}=w_{y}.
\end{equation*}

}
\end{remark}

\begin{remark}
A simple application of Hölder's inequality shows that $W_{p} \leq W_{q}$ for any $p \leq q < \infty$.
\end{remark}

We now briefly review the convergence concepts in the Wasserstein space.
The notation $\mu_{k} \rightharpoonup \mu$ means that $\mu_{k}$ converges weakly to $\mu$, i.e. $\int \varphi(x)\; \mu_{k}({\D x}) \rightarrow \int \varphi(x) \;\emph{} \mu({\D x})$ for all bounded continuous functions $\varphi:\Xc \to \Rb$.

\begin{definition}
\label{d:villani-def-1}
    %\textbf{(Weak convergence in $\left.\Pc_{p}\right)$.}
    Let $(\mathcal{X}, d)$ be a Polish space, and $p \in[1, \infty)$. Let $\left\{\mu_{k}\right\}_{k \in {N}}$ be a sequence of probability measures in $\Pc_{p}(\mathcal{X})$ and let $\mu$ be another element of $\Pc_{p}(\mathcal{X})$. Then $\left\{\mu_{k}\right\}_{k \in {N}}$ is said to converge weakly in $\Pc_{p}(\mathcal{X})$, written $\mu_k \overset{}\to \mu$,
   if any one of the following equivalent properties is satisfied for some (and then any) $x_{0} \in \mathcal{X}$:
\begin{enumerate}

\item  $\mu_{k} \rightharpoonup \mu$ and $\int d\left(x_{0}, x\right)^{p}\,  \mu_{k}({\D x}) \longrightarrow \int d\left(x_{0}, x\right)^{p} \, \mu({\D x})$,
%\item $\mu_{k} \rightharpoonup \mu$ and $\limsup _{k \rightarrow \infty} \int d\left(x_{0}, x\right)^{p} \, \mu_{k}({\D x}) \leq \int d\left(x_{0}, x\right)^{p} \, \mu({\D x})$;
\item $\mu_{k} \rightharpoonup \mu$ and $\lim _{R \rightarrow \infty} \limsup _{k \rightarrow \infty} \int_{d\left(x_{0}, x\right) \geq R} d\left(x_{0}, x\right)^{p} \, \mu_{k}({\D x})=0$,
\item For all continuous functions $\varphi$ with $|\varphi(x)| \leq 1+d\left(x_{0}, x\right)^{p}$ one has
\begin{equation*}
    \int \varphi(x) \; \mu_{k}({\D x}) \longrightarrow \int \varphi(x) \; \mu({\D x}).
\end{equation*}
    \end{enumerate}
\end{definition}

The fundamental property of the OT distance $W_p(\cdot,\cdot)$ is that it metricizes the topology of weak
convergence in $\Pc_p(\Xc)$.
\begin{theorem}
\label{t:metrization}
Let $\Xc$  be a Polish space, $p\in [1,\infty)$; then the following
two statements are equivalent:
\begin{enumerate}

\item $\mu_k \overset{}\to \mu$,
\item $W_p(\mu_k ,\mu) \to 0$.
\end{enumerate}
Furthermore, the Wasserstein space $\Pc_{p}(\mathcal{X})$, metrized by the Wasserstein distance $W_{p}$, is also a Polish space.
\end{theorem}

By the triangle inequality, we obtain the corollary.
\begin{corollary} \label{C_W}
%\textbf{(Continuity of $\left.W_{p}\right) .$}
If $(\mathcal{X}, d)$ is a Polish space, and $p \in[1, \infty)$, then $W_{p}(\cdot,\cdot)$ is continuous on $\Pc_{p}(\mathcal{X})$. More explicitly, if $\mu_{k}$ (resp. $\nu_{k}$) converges to $\mu$ (resp. $\nu$) weakly in $P_{p}(\mathcal{X})$ as $k \rightarrow \infty$, then
\begin{equation*}
W_{p}\left(\mu_{k}, \nu_{k}\right) \longrightarrow W_{p}(\mu, \nu) .
\end{equation*}
\end{corollary}

\subsection{The entropic optimal transport problem }

The OT distance is typically difficult to compute and suffers from the curse of dimensionality in empirical estimation, such that the number of samples required for reliable estimation increases exponentially with the dimension. To address these issues, regularized optimal transport methods have been introduced to enhance computational efficiency. The entropic optimal transport problem introduces a regularization term based on the Kullback-Leibler divergence to the original optimal transport problem. This regularization term encourages the optimal transport plan to be more diffuse and spread out, as opposed to the concentrated optimal transport plans that may occur in the unregularized problem. The parameter $\varepsilon$ controls the strength of the regularization, with larger values of $\varepsilon$ resulting in more diffuse optimal transport plans. The entropic optimal transport problem can be defined as follows:

\begin{definition}
For two probability measures $\mu, \nu\in \Pc_p(\Xc)$ and a regularization parameter $\varepsilon > 0$, the entropic optimal transport problem between $\mu$ and $\nu$ is defined by the formula
\begin{equation*}
W_{\varepsilon}(\mu, \nu) := \inf_{\pi \in \Pi(\mu, \nu)} \int_{\mathcal{X} \times \mathcal{X}} c(x, y) \pi({\D x}, {\D y}) + \varepsilon \mathrm{KL}(\pi | \mu \otimes \nu) 
\end{equation*}
where $\mathrm{KL}(\cdot | \cdot)$ denotes the Kullback-Leibler divergence between two probability measures, defined as
\begin{equation*}
\mathrm{KL}_(\mu  | \nu) := \int_{\mathcal{X} \times \mathcal{X}} \log\left(\frac{\D\mu}{\D\nu}\right) \D\mu,
\end{equation*}
if $\mu$ is absolutely continuous with respect to $\nu$, and $+\infty$ otherwise.
\end{definition}

To address the issue that the regularized Wasserstein distance $W_{\varepsilon}(\mu, \mu)$ is not equal to zero, Genevay et al. (2018) introduced Sinkhorn divergences, a debiased version of the regularized Wasserstein distance:
$$
\bar{W}_{\varepsilon}(\mu, \nu)=W_{\varepsilon}(\mu, \nu)-\frac{1}{2}\left(W_{\varepsilon}(\mu, \mu)+W_{\varepsilon}(\nu, \nu)\right).
$$
This normalization ensures that $\bar{W}_{\varepsilon}(\mu, \mu)=0$. The following theorems study the sample complexity and the central limit theorem of entropic optimal transport ${W}_{\varepsilon}(\mu, \mu)$, and these results can be extended to Sinkhorn divergences $\bar{W}{\varepsilon}(\mu, \nu)$ as well.

\cite[Thm 1]{Genevay_Chizat_Bach_Cuturi_Peyré_2018} quantifies the approximation error when estimating the Wasserstein distance between two probability measures using a regularized version, by providing an upper bound. 

\begin{theorem}(\cite[Thm 1]{Genevay_Chizat_Bach_Cuturi_Peyré_2018} )
Let $\mu$ and $\nu$ be probability measures on subsets $\mathcal{X}$ and $\mathcal{Y}$ of $\mathbb{R}^d$, respectively, with cardinalities bounded by $D$. Assume the cost function $c$ is L-Lipschitz with respect to both $x$ and $y$. The approximation error when estimating the Wasserstein distance $W(\mu, \nu)$ using the regularized Wasserstein distance $W_{\varepsilon}(\mu, \nu)$ is bounded by
$$
0 \leqslant W_{\varepsilon}(\mu, \nu)-W(\mu, \nu) \leqslant 2 \varepsilon d \log \left(\frac{e^2 \cdot L \cdot D}{\sqrt{d} \cdot \varepsilon}\right),
$$
which asymptotically behaves like $2 \varepsilon d \log (1 / \varepsilon)$ as $\varepsilon$ approaches 0.
\end{theorem}

\cite{Genevay_Chizat_Bach_Cuturi_Peyré_2018} also bounds the expected difference between the regularized Wasserstein distances of two measures and their empirical counterparts

\begin{theorem}(\cite[Thm 3]{Genevay_Chizat_Bach_Cuturi_Peyré_2018} )
Consider the Sinkhorn divergence between two measures $\mu$ and $\nu$ on $\mathcal{X}$ and $\mathcal{Y}$, two bounded subsets of $\mathbb{R}^d$, with a $\mathcal{C}^{\infty}$, L-Lipschitz cost c. The expected difference between the regularized Wasserstein distances of the true measures and their empirical counterparts satisfies
$$
\mathbb{E}\left|W_{\varepsilon}(\mu, \nu)-W_{\varepsilon}\left(\hat{\mu}_n, \hat{\nu}_n\right)\right|=O\left(\frac{e^{\frac{\kappa}{\varepsilon}}}{\sqrt{n}}\left(1+\frac{1}{\varepsilon^{\lfloor d / 2\rfloor}}\right)\right),
$$
where $\kappa$ depends on $|\mathcal{X}|,|\mathcal{Y}|, d$, and $||c||_{\infty}$. Here, $\hat{\mu}_n$ and $\hat{\nu}_n$ are the  empirical counterparts from $\mu$ and $\nu$, respectively.
\end{theorem}

Notably, several important statistical bound was derived for the entropic OT with the squared Euclidean cost between subgaussian probability measures in arbitrary dimension \cite{Mena_Weed_2019, CLT_sinkhorn}. 

\begin{theorem} (\cite[Thm 3.6]{CLT_sinkhorn})
Let $\mu, \nu \in \mathcal{P}\left(\mathbb{R}^d\right)$ be $\sigma$-subgaussian probability measures. Then,

$$
\sqrt{n}({W}_\varepsilon(\mu_n, \nu)-{W}_\varepsilon(\mu, \nu)) \xrightarrow{w} N(0, \operatorname{Var}_\mu(f_\varepsilon)),
$$
where $\left(f_\varepsilon, g_\varepsilon\right)$ are optimal potentials for ${W}_\varepsilon(\mu, \nu)$. Moreover, if $\lambda:=\lim_{n, m \rightarrow \infty} \frac{n}{n+m} \in(0,1)$,
$$
\sqrt{\frac{n m}{n+m}}({W}_\varepsilon(\mu_n, \nu_m)-{W}_\varepsilon(\mu, \nu)) \xrightarrow{w} N(0,(1-\lambda) \operatorname{Var}_\mu(f_\varepsilon)+\lambda \operatorname{Var}_\nu(g_\varepsilon)).
$$
\end{theorem}

One important advantage of the Central limit theorem on Sinkhorn divergences is that it can be exploited for inferential purposes and enables the build of confidence intervals.

Extending the ITD framework to incorporate entropic optimal transport presents several important advantages. This enhancement can offer a reliable and computationally efficient approximation to the original ITD, particularly in high-dimensional settings—an area where traditional optimal transport methods often struggle due to scalability issues. This is especially relevant in our use case. The ITD framework, particularly with the integration of entropic optimal transport, holds significant promise for evaluating multimodal data such as text, images, and video as well. By first applying embedding and feature extraction techniques to these diverse data types \cite{xu2025mdvt,xu2025mentor,xu2024aligngroup,lin2025camerabench,you2025multi}, we can transform them into a common representation space. Once in this space, the framework can be leveraged to compare and contrast distributions of these extracted features.

One of the key challenges in distributed or federated learning is the computational overhead and power consumption on client devices, which often have limited processing capabilities and energy constraints \cite{Zawad_Ma_Yi_Li_Zhang_Yang_Yan_He_2025,Ali_Ma_Zawad_Aditya_Akkus_Chen_Yang_Yan_2025}. As a result, minimizing resource usage while maintaining statistical accuracy becomes critical. Incorporating entropic regularization into the ITD framework can significantly reduce the computational burden, making it more practical for deployment across a large number of heterogeneous and resource-constrained clients. In addition, the entropic regularization has been proven to maintain good statistical properties in empirical settings. The well-established theoretical results for entropic optimal transport, including convergence rates and central limit theorems, could be leveraged to develop more robust confidence intervals and hypothesis tests.

Nonetheless, extending ITD to entropic optimal transport also presents significant theoretical challenges when considering the established results of the original ITD framework. The fundamental metric properties of the ITD on the space of kernels would need to be re-examined, as the addition of entropic regularization might affect these properties. The convergence behavior of sequences of kernels and measures in the entropic setting would require careful analysis, potentially leading to modified conditions for weak convergence.
These challenges highlight the need for a comprehensive theoretical analysis to bridge the gap between the established properties of entropic optimal transport and the distributed nature of the ITD framework. Such an analysis would need to carefully consider how the regularization parameter interacts with the sample sizes, the number of distributed nodes, and the underlying distributions to ensure that the desirable properties of both entropic optimal transport and ITD are preserved in the combined framework.

\bibliographystyle{unsrt}  
\bibliography{references}  
\end{document}